\documentclass[fleqn,10pt]{wlscirep}
\usepackage[utf8]{inputenc}
\usepackage[T1]{fontenc}

\usepackage{blkarray}
\usepackage{amsthm}
\newtheorem{theorem}{Theorem}

\usepackage{subfigure}
\usepackage{multirow}
\usepackage{makecell}
\title{Equivalence between Time Series Predictability and Bayes Error Rate}

\author[a,1]{En Xu}
\author[b,1,2]{Tao Zhou}
\author[a,2]{Zhiwen Yu} 
\author[a]{Zhuo Sun}
\author[a]{Bin Guo}

\affil[a]{School of computer science, Northwestern Polytechnical University, Xi'an 710129, China}
\affil[b]{Big Data Research Center, University of Electronic Science and Technology of China, Chengdu 611731, China}

\begin{abstract}
	Predictability is an emerging metric that quantifies the highest possible prediction accuracy for a given time series, being widely utilized in assessing known prediction algorithms and characterizing intrinsic regularities in human behaviors. Lately, increasing criticisms aim at the inaccuracy of the estimated predictability, caused by the original entropy-based method. In this brief report, we strictly prove that the time series predictability is equivalent to a seemingly unrelated metric called Bayes error rate that explores the lowest error rate unavoidable in classification. This proof bridges two independently developed fields, and thus each can immediately benefit from the other. For example, based on three theoretical models with known and controllable upper bounds of prediction accuracy, we show that the estimation based on Bayes error rate can largely solve the inaccuracy problem of predictability.
\end{abstract}
\begin{document}

\flushbottom
\maketitle
%
%
\thispagestyle{empty}

Predictability refers to the limit of prediction accuracy of a given time series \cite{song2010limits}. Exploring such a metric is of great significance. On the one hand, with given data, it can be used to evaluate the performance and to estimate the potential improving space of existing algorithms \cite{lu2013approaching}. On the other hand, it characterizes intrinsic regularities of time series and thus the investigation can deepen our understanding of related phenomena. For example, the varying predictability reveals the sudden change of human mobility patterns after disasters \cite{lu2012predictability}, quantifies the respective contributions of behavioral similarities and social relationships to human mobility prediction \cite{Chen2022}, and uncovers the roles of model structure and social network heterogeneity in predicting infectious disease outbreaks \cite{Scarpino2019}.

Song \textit{et al.} proposed an entropic framework to calculate the predictability $\Pi$ \cite{song2010limits} (see also an analogous contribution by Feder and Merhav \cite{feder1994relations}). This pioneering method has two obvious disadvantages: (i) it builds on the underlying Markovian nature while some time series exhibit long-range correlations; (ii) it is sensitive to the lengths of time series while real-time series are usually too short to satisfy the requirement. Consequently, the above method is usually inaccuracy, sometimes largely overestimated \cite{xu2019predictability,smith2014refined} and sometimes surpassed by well-designed algorithms \cite{lu2013approaching,kulkarni2019examining}.

Bayes error rate (BER, denoted by $R$) is the lowest unavoidable error rate in classification for given data \cite{cover1967nearest}, which has been extensively applied in feature selection, intermediate representations of features or behaviors, quality assessment of security defenses, feasibility estimation of machine learning, and so on \cite{berisha2015empirically}. In contrast to the onset of germination of predictability, the investigation on BER has lasted more than half a century and thus many sophisticated methods are proposed to calculate $R$ or to estimate the upper and lower bounds of $R$, such as the \textit{density estimators} that directly estimate the conditional probability distribution $\eta_{\omega}(x)=p(\omega|x)$ where $\omega$ and $x$ are class label and feature vector \cite{ferguson1983bayesian}, the \textit{divergence estimators} that focus on the maximum posterior probability $\max \limits_{\omega \in \Omega} \eta_{\omega}(x)$ where $\Omega$ is the set of classes \cite{nguyen2010estimating}, and the \textit{k-NN estimators} that learn the expectation of the highest accuracy $\mathbb{E}_{X}[1-\max \limits_{\omega \in \Omega} \eta_{\omega}(x)]$, where $X$ denotes the feature space \cite{fukunaga1975k}. Later in this brief report, we will prove the mathematical equivalence between the above two seemingly unrelated metrics, predictability and BER, and further discuss what can be immediately gained after the equivalence.

\begin{table}
	\centering
	\caption{Illustration of the correspondence between prediction and classification for an example series ABCBA. } 
	\scalebox{0.85}{
		\begin{tabular}{cllll} 
			\toprule
			\multirow{2}{*}{\textbf{Time Series}} & \multicolumn{2}{c}{\textbf{Prediction}}&\multicolumn{2}{c}{\textbf{Classification}}   \\
			\cline{2-5} 
			& \textbf{Historical States} & \textbf{Predicted State} &\textbf{Features} & \textbf{Class}  \\
			\midrule
			\multirow{5}{*}{ABCBA} 
			& 1.  $\varnothing$& 1.  A	& 1.  $\varnothing$& 1.  A	\\
			& 2.  A		& 2.  B	& 2.  A		& 2.  B			\\
			& 3.  AB	& 3.  C	& 3.  AB	& 3.  C					\\
			& 4.  ABC	& 4.  B	& 4.  ABC	& 4.  B					\\
			& 5.  ABCB  & 5.  A     & 5.  ABCB     &5.  A			\\
			\bottomrule
	\end{tabular}}
	\label{tbl:1}
\end{table}

\section*{Theorem}

\begin{theorem}
	\label{theorem:1}
	Given a M-state time series, its predictability $\Pi$ is equivalent to the Bayes error rate $R$ in a M-classification problem as
	\begin{equation}
		\Pi = 1 - R, \label{eqn:1}
	\end{equation}	
	if we treat each state as a class, and the series before the state as the feature. 
\end{theorem}

\begin{proof}
	Denote $x_{n-1}=\omega^{1}\omega^{2}\cdots\omega^{n-1}$ the historical series from time $1$ to $n-1$, where $\omega^{i}\in \Omega$ and $\Omega$ is the set of $M$ states. Denote $Pr[\omega^{n}=\hat{\omega}^{n}|x_{n-1}]$ the probability that our predicted state $\hat{\omega}^{n}$ is equal to the actual state $\omega^{n}$ given $x_{n-1}$, and $\pi(x_{n-1})=\sup_{\omega}\left\{Pr[\omega^{n}=\omega|x_{n-1}] \right\}$ the probability of occurrence of the most probable state at time $n$, then the predictability of the $n$th state given the historical states $x_{n-1}$ is $\Pi(n)=\sum_{x_{n-1}}P(x_{n-1})\pi(x_{n-1})$, where $P(x_{n-1})$ is the probability of observing a particular history $x_{n-1}$, and the sum is taken over all possible histories of length $n-1$. Notice that, $\pi(x_{n-1})$ contains the full predictive power including the potential long-range correlations in the time series, while in practice we usually use shorter historical series, such as  $\omega^{n-r}\omega^{n-r+1}\cdots\omega^{n-1}$ with $r$ a cutoff parameter, instead of the full historical series $\omega^{1}\omega^{2}\cdots\omega^{n-1}$, so that in a more general case, the probability of observing a particular history can be smaller than 1. The overall predictability $\Pi$ is then defined as the time averaged predictability for a sufficiently long time series \cite{song2010limits}, as 
	\begin{align}
		\Pi = \lim\limits_{n\to \infty}\frac{1}{n}\sum_{i=1}^{n}\Pi(i),  
	\end{align}
	where $x_0=\varnothing$ and $\Pi(1)$ is the predictability of the first state without any available information.

	Considering an $M$-classification problem with $n$ samples whose class labels are the $n$ states in the time series and whose features are series before the corresponding states. Table \ref{tbl:1} illustrates the one-to-one relationship between time series prediction and classification for an example series with $M=3$ and $n=5$. Denote $p\left(x | \omega_{j}\right)$ the conditional probability density of the feature $x \in X'$, $X'\subset X$ is the set of observed features, and $p\left(\omega_{j}\right)$ the prior probability of the class $\omega_{j}\in \Omega$ $(j = 1, 2, \cdots, M)$, the BER is expressed as:
	\begin{equation}
		R=1-\sum_{j=1}^{M} \int_{\Gamma_{j}} p\left(\omega_{j}\right) p\left(x | \omega_{j}\right) dx,
	\end{equation}
	with the partition $\Gamma_{j}$ defined as:
	\begin{equation}
		\Gamma_{j} \!\triangleq\! \left\{\! x\! \in\! X'\! \mid\! p\left(\omega_{j}\right) p\left(x | \omega_{j}\right)\!>\!\max _{\substack{k \neq j}}\left\{p\left(\omega_{k}\right) p\left(x | \omega_{k}\right)\right\}\! \right\}.
	\end{equation}
	
	Applying the Bayes formula $p(x)p\left(\omega_{j}|x\right) = p\left(\omega_j\right)p\left(x | \omega_{j}\right)$, where $p(x)$ is the prior probability of the feature $x$, we have
	\begin{equation}
		\label{eqn:5}
		\sum_{j=1}^{M} \int_{\Gamma_{j}} p\left(\omega_{j}\right) p\left(x | \omega_{j}\right) dx = \sum_{j=1}^{M} \int_{\Gamma_{j}^{'}} p\left(x\right) p\left(\omega_{j} | x\right) dx,
	\end{equation}
	and the partition $\Gamma_{j}$ is equivalent to the partition $\Gamma_{j}^{'}$ with
	\begin{equation}
		\Gamma^{'}_{j} \triangleq \left\{ x \in X' \mid p\left(\omega_{j}|x \right)>\max _{\substack{k \neq j}}\left\{p\left(\omega_{k}|x\right)\right\} \right\}.
	\end{equation}
	
	According to our setting, there is a one-to-one relationship between features and historical series, namely for each $x_{i-1}$ $(1 \leq i \leq n)$, there exist a certain $x \in X'$ s.t. $x_{i-1}=x$, and vice versa. As $p(x)$ is the probability of observing the feature $x$ in the feature set $X'$ while $P(x_{i-1})$ is the probability of observing the series $x_{i-1}$ in all historical series of length $i-1$, in the large limit of $n$, $p(x)=\frac{1}{n}P(x_{i-1})$. In addition, $\pi(x_{i-1})=p\left(\omega_{j}|x\in \Gamma^{'}_{j} \right)$ according to the definition of $\pi(\cdot)$. As a consequence,  
	\begin{equation}
			\sum\limits_{j=1}^{M} \int_{\Gamma^{'}_{j}} p(x) p\left(\omega_{j} |x\right) dx = \lim\limits_{n\to \infty}\frac{1}{n}\sum\limits_{i=1}^{n}\sum\limits_{x_{i-1}}P(x_{i-1})\pi(x_{i-1}).
	\end{equation}
	
	Therefore, the time series predictability is equivalent to the Bayes error rate, with the relationship $\Pi=1-R$. 
\end{proof}

\section*{Results}

According to Theorem \ref{theorem:1}, we can directly take advantage of methods developed to calculate $R$ in real datasets to improve the estimation of $\Pi$. Considering a simple example with three states $\Omega = \{A,B,C\}$, where the next state only depends on the current state, according to the following Markovian transfer matrix 
\begin{equation}
	\begin{blockarray}{cccc}
		&A&B&C\\
		\begin{block}{c[ccc]}
			A&	q&	\frac{2}{3}(1-q)&	\frac{1}{3}(1-q)\\
			B&	\frac{1}{3}(1-q)&	q&	\frac{2}{3}(1-q)\\
			C&	\frac{2}{3}(1-q)&	\frac{1}{3}(1-q)&	q\\
		\end{block}
	\end{blockarray}.
\end{equation}
Obviously, when $0.4 \le q \le 1$, the true predictability $T=q$. Time series with arbitrary length $n$ can be generated by Eq. 8. We set $r=1$ to extract the features. Take $\{ABBCA\cdots\}$ as an example, the corresponding (feature, class) set is $\{(A,B),(B,B),(B,C),(C,A),\cdots\}$.

According to the entropy-based method, the estimated predictability $\bar{\Pi}$ is determined by 
\begin{equation}
	\label{eqn:9}
	H = -\bar{\Pi}\log_2\bar{\Pi} - (1-\bar{\Pi})\log_2(1-\bar{\Pi}) + (1-\bar{\Pi})\log_2(M-1),
\end{equation}
where $M=3$ and $H$ is the entropy of the next-moment state that can be estimated by the data (see details in \cite{song2010limits}). In the corresponding classification problem, the lower and upper bounds of $R$ can be obtained by the inequality
\begin{equation}
	\label{eqn:10}
	\begin{gathered}
		\frac{M-1}{(M-2) M} \sum_{i=1}^{M}\left[1-p(\omega_i)\right] R_{i}^{M-1} \leq R^{M} \leq 
		\min_{\alpha \in\{0,1\}} \frac{1}{M-2 \alpha} \sum_{i=1}^{M}\left[1-p(\omega_i)\right] R_{i}^{M-1}+\frac{1-\alpha}{M-2 \alpha},
	\end{gathered}
\end{equation}
where $R^k$ is the BER for the $k$-classification subproblem and $R_{i}^{M-1}$ is the BER for the $(M-1)$-classification subproblem created by removing the $i$th class (see details in \cite{wisler2016empirically,renggli2021evaluating}). The upper and lower bounds of predictability can then be obtained through Eq. \ref{eqn:1} (Theorem \ref{theorem:1}), and the estimated predictability $\tilde{\Pi}$ is the average of the two bounds. 

Figure 1A shows how $\bar{\Pi}$ changes with increasing $n$ for three specific cases $q=0.4$, $q=0.6$ and $q=0.8$. The result confirms two above-mentioned disadvantages of the entropy-based method, namely $\bar{\Pi}$ is sensitive to the length $n$ and much larger than the true predictability $T=q$. To ensure the stability, we set $n=2^{15}$ and compare the entropy-based method (Eq. \ref{eqn:9}) and the BER-inspired method (Eq. \ref{eqn:10}). As shown in figure 1B, the latter remarkably and consistently outperforms the former.  

Considering a more complicated series generator with $M$ states $\Omega = \{S_1,S_2,\cdots,S_M\}$, where the next state $\omega^{t+1}$ is randomly drawn from the $M$ states with probability $1-q$, or determined by the two anterior states with probability $q$. In the latter case, if $\omega^{t-1}=S_i$ and $\omega^t=S_j$, then $\omega^{t+1}=S_k$, $k=i+j$ (if $k>M$, we set $k \leftarrow k-M$). Obviously, the true predictability is $T=q+(1-q)/M$. Figure 1C reports how $\bar{\Pi}$ changes with increasing $n$ for four specific cases $q=0.2$, $q=0.4$, $q=0.6$ and $q=0.8$, with $M=100$ fixed. As $1/M$ is much smaller than $q$ in the above four cases, $T \approx q$.   Analogous to what found in figure 1A, $\bar{\Pi}$ is sensitive to $n$ and much larger than $T$ after being nearly stable ($n>2^{15}$). As shown in figure 1D, in most cases the BER-inspired method performs better than the entropy-based method, and only when the time series is highly predictable ($q\approx 1$, see the top right corner), the results of the entropy-based method and BER-inspired method are close to each other.

To reveal the effects of parameters $r$ and $M$, we further consider the third generator where the next state $\omega^{t+1}$ is equal to $\omega^{t}$ with probability $q_{1}=0.1$, equal to $\omega^{t-1}$ with probability $q_{2}=0.2$, equal to $\omega^{t-2}$ with probability $q_{3}=0.3$. With probability $1-q_1-q_2-q_3$, $\omega^{t+1}$ is randomly drawn from $M$ states. The true predictability is $T=\max\{q_1,\cdots,q_r\}+(1-q_1-q_2-q_3)/M$, sensitive to $r$ and $M$. As shown in figure 1E, the original entropy-based method does not consider the impacts of parameter $r$ while the BER-inspired method can well capture the effects of the memory length $r$. As shown in figure 1F, both the entropy-based and BER-inspired methods capture the decreasing tendency of predictability as the increase of $M$. One can clearly observed from figures 1E and 1F that the entropy-based method will largely overestimate the predictability even for sufficiently long time series, while the BER-inspired method performs much better.

\begin{figure*}[t]
	\centering
	\subfigure{
		\label{fig:side:a}
		\includegraphics[width=2in]{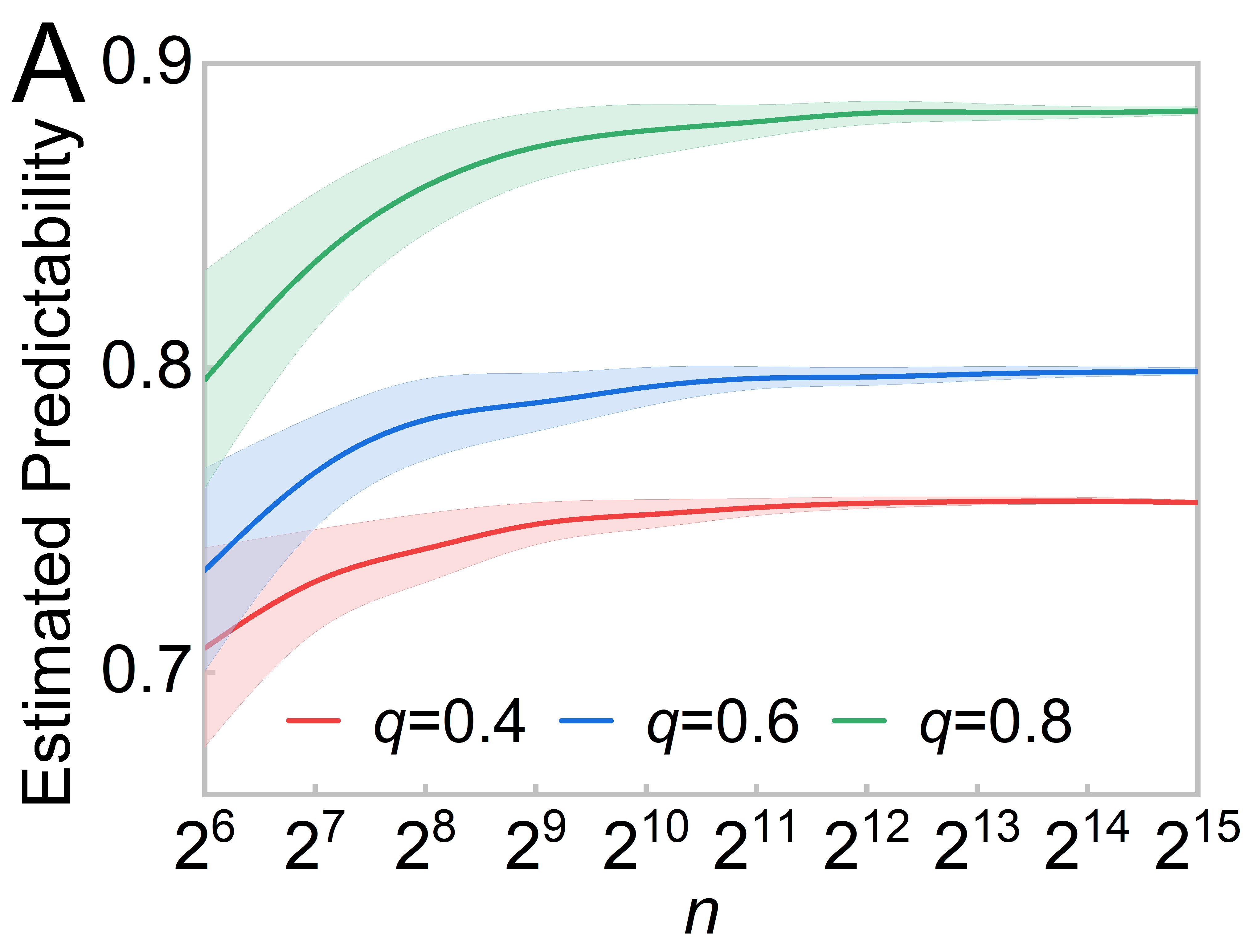}
	}
	\subfigure{
		\label{fig:side:b}
		\includegraphics[width=2in]{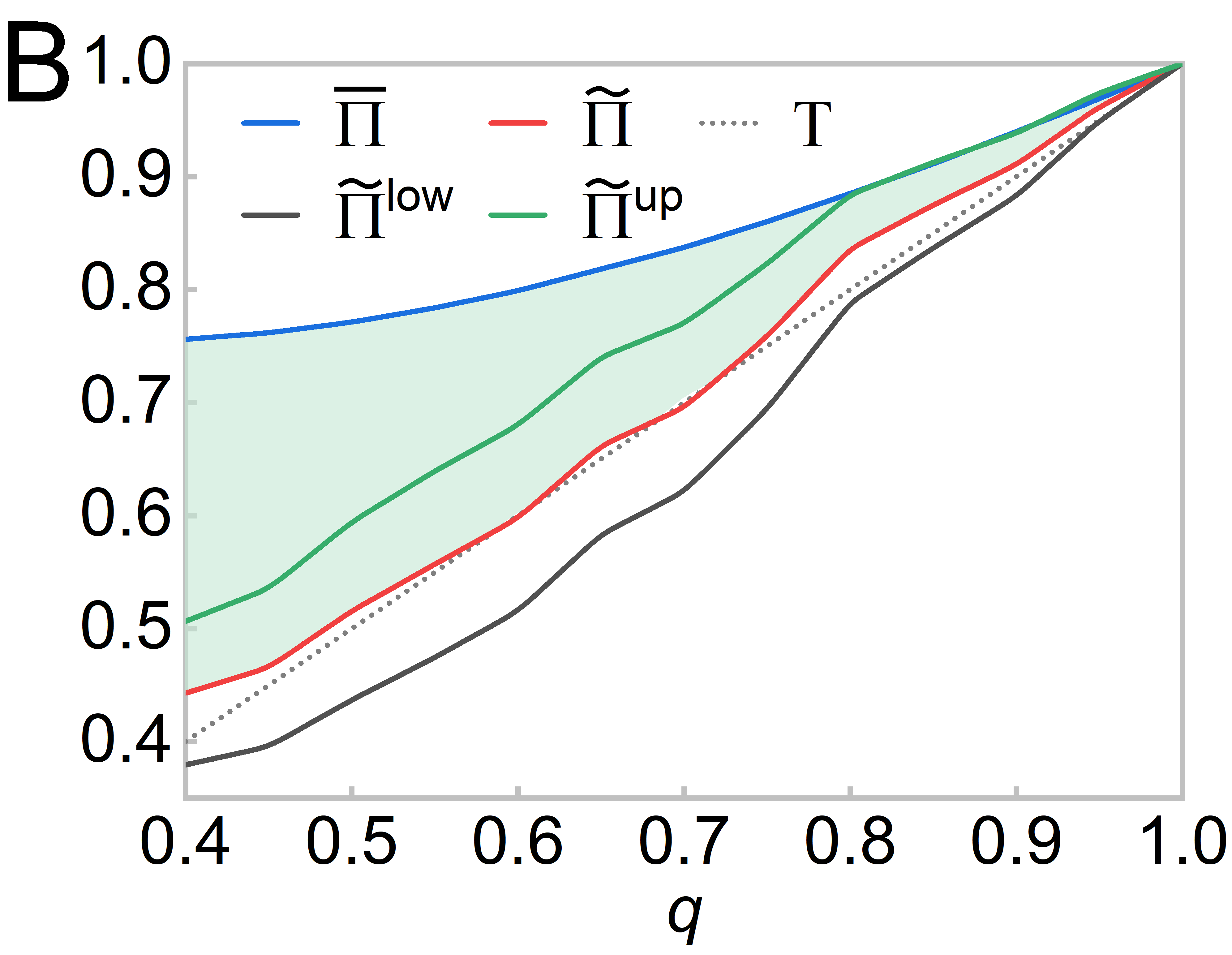}
	}
	\subfigure{
		\label{fig:side:c}
		\includegraphics[width=2in]{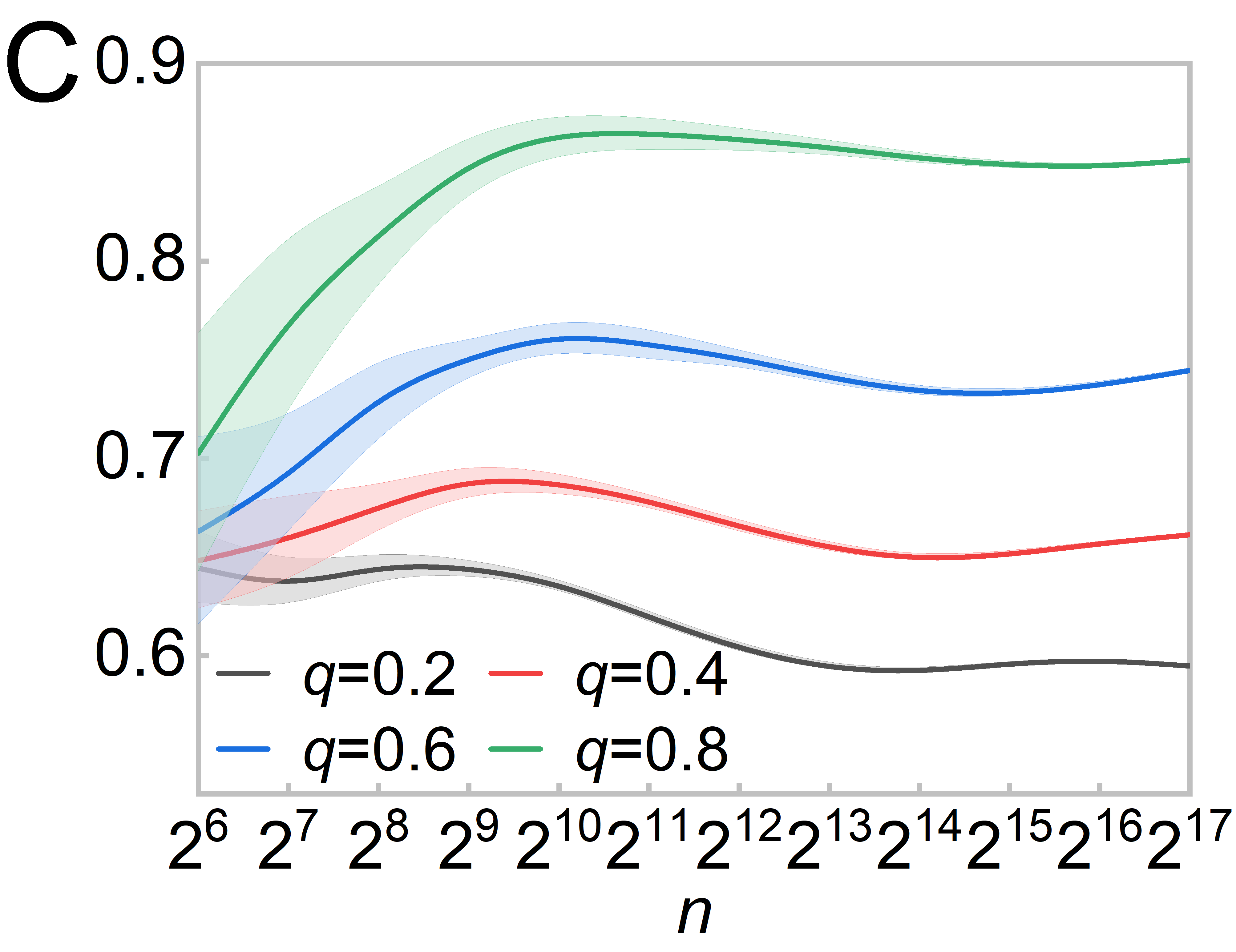}
	}
	\subfigure{
		\label{fig:side:d}
		\includegraphics[width=2.08in]{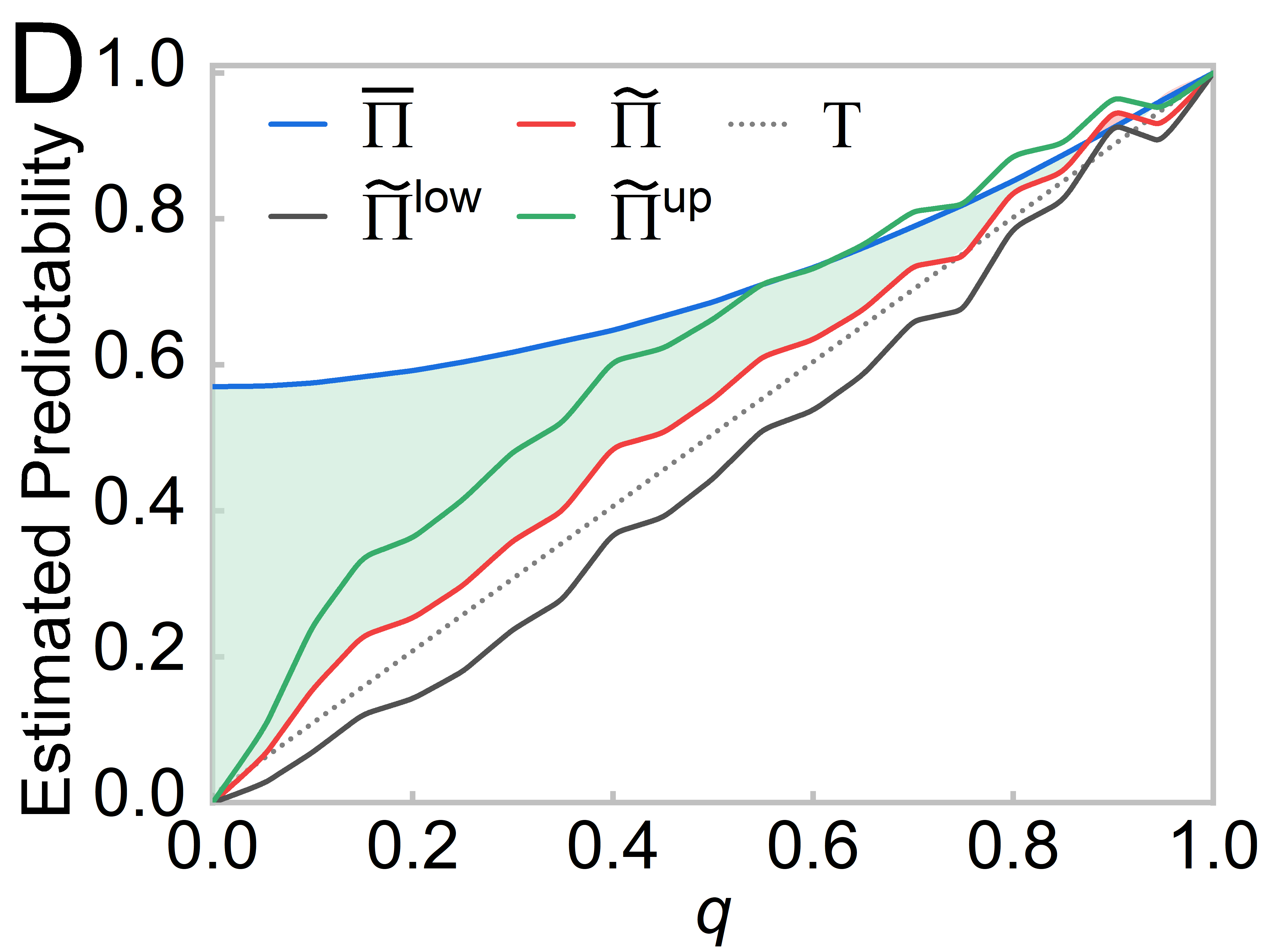}
	}
	\subfigure{
		\label{fig:side:e}
		\includegraphics[width=2in]{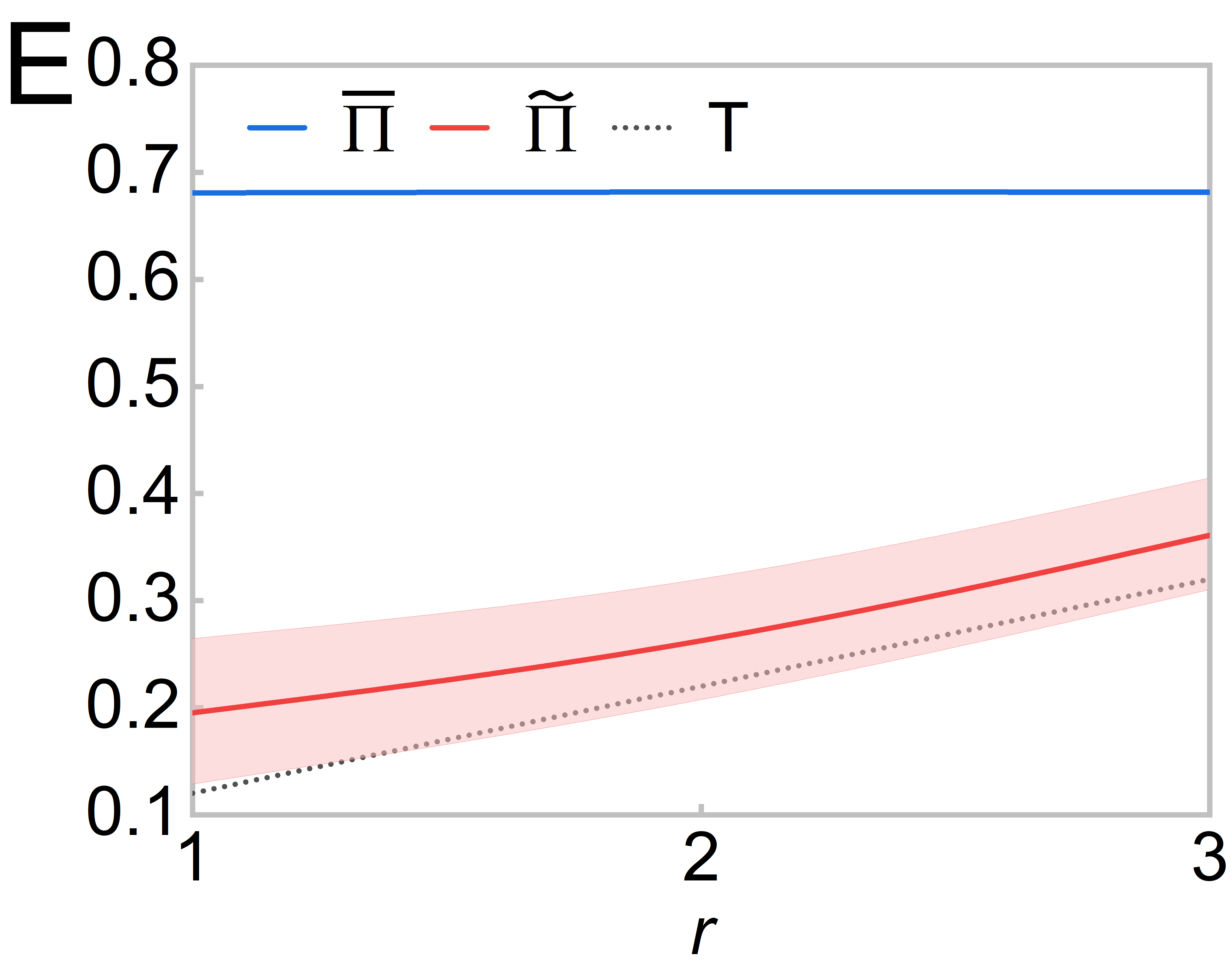}
	}
	\subfigure{
		\label{fig:side:f}
		\includegraphics[width=2.07in]{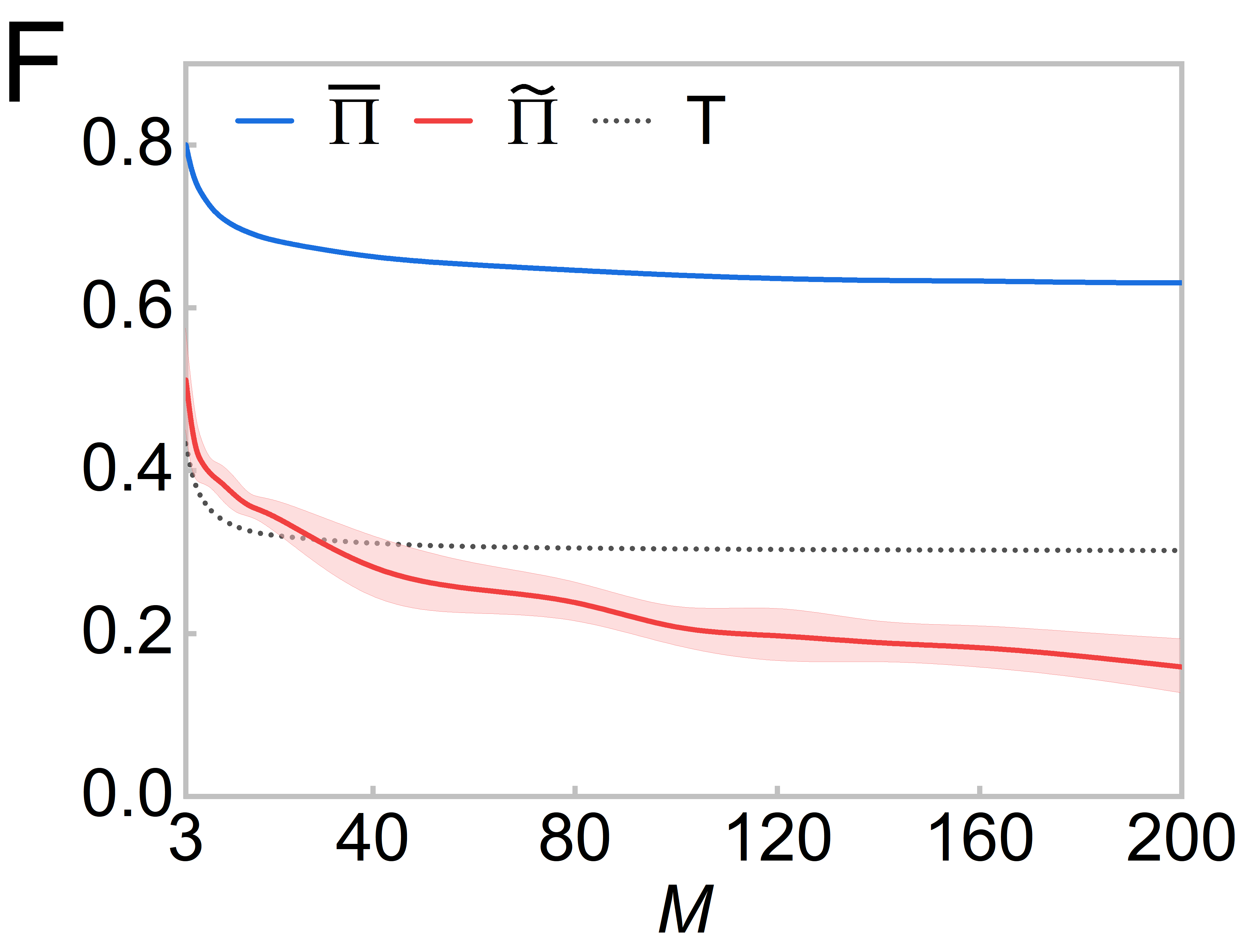}
	}
	\vspace{0pt}
	\caption{(A) How the estimated predictability $\bar{\Pi}$ by the entropy-based method changes with the increasing $n$ under the first series generator. (B) The performance of the entropy-based method (Eq. 9) and the BER-inspired method (Eq. 10) under the first generator. (C) How the estimated predictability $\bar{\Pi}$ by the entropy-based method changes with the increasing $n$ under the second series generator. (D) The performance of the entropy-based and BER-inspired methods under the second generator. For the BER-inspired method, $\tilde{\Pi}^{\textup{low}}$ and $\tilde{\Pi}^{\textup{up}}$ are the lower and upper bounds by Eq. (10), and  $\tilde{\Pi}=\frac{1}{2}\left( \tilde{\Pi}^{\textup{low}}+\tilde{\Pi}^{\textup{up}} \right)$ is the estimated predictability. In plots (B) and (D), the shadow areas indicate to what extent the BER-inspired method outperforms the entropy-based method. (E) The performance of the entropy-based and BER-inspired methods under the thrid generator with varying $r$, with $M=20$ fixed. (F) The performance of the entropy-based and BER-inspired methods under the thrid generator with varying $M$, with $r=3$ fixed. In plots (E) and (F), the shadow areas denote the standard errors. In all comparisons between the entropy-based and BER-inspired methods, the length of time series is fixed as $n=2^{15}$, and the corresponding results are averaged over 10 independent runs.}
	\label{fig:results}
\end{figure*}

\section*{Discussion}

The direct value of knowing predictability is to decide whether it is worthwhile to improve the current predictors \cite{song2010limits,lu2015toward}. The embodiment of such value requires an accurate estimate of predictability. Unfortunately, the entropy-based method \cite{song2010limits} usually fails as it largely overestimates the true predictability (see, for example, figure \ref{fig:results}). The dissatisfactory performance partially comes from the approximation that only accounts for the entropy of the state with the maximum next-moment occurrence probability. At the same time, such approximation is an indispensable part that guarantees the computational feasibility. Therefore, it is difficult to overcome the observed disadvantages within the entropic framework \cite{smith2014refined,zhang2022beyond}. This paper uncovers the equivalence between predictability and a seemingly unrelated metric BER, and immediately provides  
a novel way to improve the estimation of predictability -- applying the BER-inspired methods.

\section*{Acknowledgement}

This work was supported in part by the National Natural Science Foundation of China (No. 61960206008, No. 62002294, No. 11975071) and the National Science Fund for Distinguished Young Scholars (No. 61725205).

\bibliography{sample}

\section*{Author contributions statement}

E.X., Z.Y., B.G., and L.Y. designed research; E.X. performed research; and E.X., Z.Y., B.G., and L.Y. wrote the paper.
Author contributions: E.X., T.Z., and Z.Y. designed research; E.X. and T.Z. performed research; E.X. and T.Z. proved the theorem; E.X., T.Z., Z.Y., Z.S., and B.G. analyzed results; E.X. and T.Z. wrote the paper; Z.Y., Z.S., and B.G. edited the paper.

\textsuperscript{1}E.X. and T.Z. contributed equally to this work.

\textsuperscript{2}To whom correspondence should be addressed. E-mail: zhutou@ustc.edu (T.Z.) or zhiwenyu@nwpu.edu.cn (Z.Y.)

\end{document}